\documentclass{llncs}
\usepackage{alltt}
\usepackage[pdftex]{graphicx}
\usepackage{stmaryrd}
\usepackage{amssymb}
\newcommand\True{\ensuremath{\mathit{True}}}
\newcommand\False{\ensuremath{\mathit{False}}}
\newcommand\TT{\ensuremath{\mathit{tt}}}
\newcommand\FF{\ensuremath{\mathit{ff}}}
\newcommand\Not{\ensuremath{\mathit{Not}}}
\newcommand\Copy{\ensuremath{\mathit{Copy}}}
\newcommand\And{\ensuremath{\mathit{And}}}

\newcommand{\loves}{\ensuremath{\vdash}}

\newcommand\unknown{\ensuremath{\mathit{unknown}}}
\newcommand\fv[1]{\ensuremath{\mathbf{fv}(#1)}}
\newcommand\lab[1]{\ensuremath{\mathbf{lab}(#1)}}
\newcommand\dom[1]{\ensuremath{\mathbf{dom}(#1)}}
\newcommand\ev[1]{\ensuremath{\mathcal{E}\sem{#1}}}

\newcommand\av[1]{\ensuremath{\mathcal{A}\sem{#1}}}
\newcommand\sem[1]{\ensuremath{\llbracket #1 \rrbracket}}
\newcommand\nc{{\sc nc}}
\newcommand\ptime{{\sc ptime}}
\newcommand\logspace{{\sc logspace}}
\newcommand\cache{\widehat{\mathsf{C}}}

\newcommand\Lab{\mathbf{Lab}}

\newcommand\Term{\mathbf{Term}}

\newcommand\Exp{\mathbf{Exp}}
\newcommand\Env{\mathbf{Env}}

\newcommand\Var{\mathbf{Var}}

\newcommand\restrict{\ensuremath{\!\upharpoonright\!}}

\title{Flow analysis, linearity, and PTIME}
\author{David {Van Horn} \and Harry G.~Mairson}
\institute{
Department of Computer Science\\
Brandeis University\\
Waltham, Massachusetts 02454\\
\email{\{dvanhorn,mairson\}@cs.brandeis.edu}
}

\begin{document}

\mainmatter
\maketitle

\begin{abstract}

Flow analysis is a ubiquitous and much-studied component of compiler
technology---and its variations abound.  Amongst the most well known is
Shivers' 0CFA; however, the best known algorithm for 0CFA requires
time cubic in the size of the analyzed program and is unlikely to be
improved.  Consequently, several analyses have been designed to
approximate 0CFA by trading precision for faster computation.
Henglein's simple closure analysis, for example, forfeits the notion
of directionality in flows and enjoys an ``almost linear'' time
algorithm.  But in making trade-offs between precision and complexity,
what has been given up and what has been gained?  Where do these
analyses differ and where do they coincide?

We identify a core language---the linear $\lambda$-calculus---where
0CFA, simple closure analysis, and many other known approximations or
restrictions to 0CFA are rendered identical.  Moreover, for this core
language, analysis corresponds with (instrumented) evaluation.
Because analysis faithfully captures evaluation, and because the
linear $\lambda$-calculus is complete for \ptime, we derive
\ptime-completeness results for all of these analyses.
\end{abstract}



\section{Introduction}
\label{sec:intro}

Flow analysis \cite{jones-81,sestoft-88,shivers-phd,midtgaard-07} is
concerned with providing sound approximations to the question of
``does a given value flow into a given program point during
evaluation?''  The most approximate analysis will always answer {\em
yes}, which takes no resources to compute---and is of little use.  On
the other hand, the most precise analysis will answer {\em yes} if and
only if the given value flows into the program point during
evaluation, which is useful, albeit uncomputable.  In mediating
between these extremes, every static analysis necessarily gives up
valuable information for the sake of computing an answer within
bounded resources.  Designing a static analyzer, therefore, requires
making trade-offs between precision and complexity. But what exactly
is the analytic relationship between forfeited information and
resource usage for any given design decision?  In other words:
\begin{center}
{\em What are the computationally potent ingredients in a
static analysis?}
\end{center}

The best known algorithms to compute Shivers' 0CFA \cite{shivers-88},
a canonical flow analysis for higher-order programs, are cubic in the
size of the program, and there is strong evidence to suggest that in
general, this cannot be improved \cite{heintze-mcallester-97b}.
Nonetheless, information can be given up in the service of quickly
computing a necessarily less precise analysis, avoiding the ``cubic
bottleneck.''  For example, by forfeiting 0CFA's notion of
directionality, algorithms for Henglein's simple closure analysis run
in near linear time \cite{henglein92d}.  Similarly, by explicitly
bounding the number of passes the analyzer is allowed over the
program, as in Ashley and Dybvig's sub-0CFA \cite{ashley-dybvig}, we
can recover running times that are linear in the size of the program.
But the question remains: Can we do better?  For example, is it
possible to compute these less precise analyses in logarithmic space?
We show that without profound combinatorial breakthroughs (\ptime\ $=$
\logspace), the answer is no.  Simple closure analysis, sub-0CFA, and
other analyses that approximate or restrict 0CFA, {\em require}---and
are therefore, like 0CFA \cite{vanhorn-mairson-07}, complete
for---polynomial time.

\paragraph{What is flow analysis?}

Flow analysis is {\em the} ubiquitous static analysis of higher-order
programs.  As Heintze and McAllester remark, some form of flow
analysis is used in most forms of analysis for higher-order languages
\cite{heintze-mcallester-97a}.  It answers fundamental questions such
as {\em what values can a given subexpression possibly evaluate to at
run-time?}  Each subexpression is identified with a unique
superscripted {\em label} $\ell$, which serves to index that program
point.  The result of a flow analysis is a {\em cache} $\cache$ that
maps program points (and variables) to sets of values.  These analyses
are {\em conservative} in the following sense: if $v$ is in
$\cache(\ell)$, then the subexpression label $\ell$ {\em may} evaluate
to $v$ when the program is run (likewise, if $v$ is in $\cache(x)$,
$x$ may be bound to $v$ when the program is run).  But if $v$ is {\em
not} in $\cache(\ell)$, we know that $e$ {\em cannot} evaluate to $v$
and conversely if $e$ evaluates to $v$, $v$ {\em must} be in
$\cache(\ell)$.

For the purposes of this paper and all of the analyses considered
herein, values are (possibly open) $\lambda$-abstractions.  During
evaluation, functional values are denoted with {\em closures}---a
$\lambda$-abstraction together with an {\em environment} that closes
it.  Values considered in the analysis approximate run-time
denotations in the sense that if a subexpression labeled $\ell$
evaluates to the closure $\langle \lambda x.e,\rho\rangle$, then
$\lambda x.e$ is in $\cache(\ell)$.

The {\em acceptability} of a flow analysis is often specified as a set
of (in)equations on program fragments.  The most naive way to compute
a satisfying analysis is to initialize the cache with the flow sets
being empty.  Successive passes are then made over the program,
monotonically updating the cache as needed, until the least fixed
point is reached.  The more approximate the analysis, the faster this
algorithm converges on a fixed point.  The key to a fruitful analysis,
then, is ``to accelerate the analysis without losing too much
information'' \cite{ashley-dybvig}.

One way to realize the computational potency of a static analysis is
to subvert this loss of information, making the analysis an {\em
exact} computational tool.  Proving lower bounds on the expressiveness
of an analysis becomes an exercise in hacking, armed with this
newfound tool.  Clearly the more approximate the analysis, the less
there is to work with, computationally speaking, and the more we have
to do to undermine the approximation.  But a fundamental technique
allows us to understand expressivity in static analysis: {\em
linearity}.  This paper serves to demonstrate that linearity renders a
number of the most approximate flow analyses both equivalent and
exact.

\paragraph{Linearity and approximation in static analysis:}

Linearity, the Curry-Howard programming counterpart of linear logic
\cite{girard-ll}, plays an important role in understanding static
analyses.  The reason is straightforward: because static analysis has
to be tractable, it typically approximates normalization, instead of
simulating it, because running the program may take too long.  For
example, in the analysis of simple types---surely a kind of static
analysis---the approximation is that all occurrences of a bound
variable must have the same type (as a consequence, perfectly good
programs are rejected).  A comparable but not identical thing happens
in the case of type inference for ML and bounded-rank intersection
types---but note that when the program is linear, there is no
approximation, and type inference becomes evaluation under another
name.

In the case of flow analysis, similarly, a cache is computed in the
course of an approximate evaluation, which is only an approximation
because each evaluation of an abstraction body causes the same piece
of the cache to be (monotonically) updated.  Again, if the term is
linear, then there is only one evaluation of the abstraction body, and
the flow analysis becomes synonymous with normalization.

\paragraph{Organization of this paper:}

The next section introduces 0CFA in order to provide intuitions and a
point of reference for comparisons with subsequent analyses.  Section
\ref{sec:simple} specifies and provides an algorithm for computing
Henglein's simple closure analysis. Section \ref{sec:linearity}
develops a correspondence between evaluation and analysis for linear
programs.  We show that when the program is linear, normalization
and analysis are synonymous.  As a consequence the normal form of a
program can be {\em read back} from the analysis.  We then show in
Section \ref{sec:circuits} how to simulate circuits, the canonical
\ptime-hard problem, using linear terms.  This establishes the
\ptime-hardness of the analysis.  Finally, we discuss other
monovariant flow analyses and sketch why these analyses remain
complete for \ptime\ and provide conclusions and perspective.

\section{Shivers' 0CFA}
\label{sec:0cfa}

As a starting point, we consider Shivers' 0CFA
\cite{shivers-88,shivers-phd}.\footnote{It should be noted that
Shivers' original {\em zeroth-order control-flow analysis} (0CFA) was
developed for a core CPS-Scheme language, whereas the analysis
considered here is in direct-style.  Sestoft independently developed a
similar flow analysis in his work on globalization
\cite{sestoft-88,sestoft-89}.  See \cite{midtgaard-07,mossin-phd} for
details.}

\paragraph{The Language:} 
A countably infinite set of labels, which include variable names, is
assumed.  The syntax of the language is given by the following
grammar:
\begin{displaymath}
\begin{array}{l@{\hspace{1cm}}l@{\hspace{1cm}}l}
\Exp & e ::= t^\ell & \mbox{expressions (or labeled terms)}\\
\Term & t ::= x\ |\ e\;e\ |\ \lambda x.e & \mbox{terms (or unlabeled expressions)}
\end{array}
\end{displaymath}

All of the syntactic categories are implicitly understood to be
restricted to the finite set of terms, labels, variables, etc.~that
occur in the {\em program of interest}---the program being analyzed.
The set of labels, which includes variable names, in a program
fragment is denoted $\lab{e}$.  As a convention, programs are assumed
to have distinct bound variable names.

The result of 0CFA is an {\em abstract cache} that maps each program
point (i.e., label) to a set of $\lambda$-abstractions which
potentially flow into this program point at run-time:
\begin{displaymath}
\begin{array}{ccll}
\cache & \in & \Lab \rightarrow
\mathcal{P}(\Term) & \qquad\mbox{abstract caches}
\end{array}
\end{displaymath}
Caches are extended using the notation $\cache[\ell \mapsto s]$, and
we write $\cache[\ell \mapsto^+ s]$ to mean $\cache[\ell \mapsto
(s\cup \cache(\ell))]$.  It is convenient to sometimes think of caches
as mutable tables (as we do in the algorithm below), so we abuse
syntax, letting this notation mean both functional extension and
destructive update.  It should be clear from context which is implied.

\paragraph{The Analysis:} We present the specification of the analysis here in the style of
Nielson {\em et al.} \cite{nielson-nielson-hankin}.  Each
subexpression is identified with a unique superscripted {\em label}
$\ell$, which marks that program point; $\cache(\ell)$ stores all
possible values flowing to point $\ell$.  An {\em acceptable} flow
analysis for an expression $e$ is written $\cache\models e$:
\begin{eqnarray*}
\cache &\models & 
x^\ell \mbox{ iff } \cache(x) \subseteq \cache(\ell)\\
\cache &\models & 
(\lambda x .e)^\ell \mbox{ iff } \lambda x.e \in \cache(\ell)\\
\cache &\models & 
(t_1^{\ell_1}\ t_2^{\ell_2})^\ell \mbox{ iff } \cache\models t_1^{\ell_1} \wedge
\cache\models t_2^{\ell_2} \wedge \forall \lambda x.t_0^{\ell_0} \in \cache(\ell_1) : \\
\ & \ & \quad
  \cache\models t_0^{\ell_0} \wedge \cache(\ell_2) \subseteq \cache(x) \wedge
  \cache(\ell_0) \subseteq \cache(\ell)
\end{eqnarray*}
The $\models$ relation needs to be coinductively defined since
verifying a judgment $\cache \models e$ may obligate verification of
$\cache \models e'$ which in turn may require verification of $\cache
\models e$.  The above specification of acceptability, when read as a
table, defines a functional, which is monotonic,
has a fixed point, and $\models$ is defined coinductively as the
greatest fixed point of this functional.\footnote{See
\cite{nielson-nielson-hankin} for a thorough discussion of coinduction
in specifying static analyses.}

Writing $\cache \models t^\ell$ means ``the abstract cache $\cache$
contains all the flow information for program fragment $t$ at program
point $\ell$.''  The goal is to determine the {\em least} cache
solving these constraints to obtain the most precise analysis. Caches
are partially ordered with respect to the program of interest:
\begin{displaymath}
\cache \sqsubseteq \cache' \mbox{ iff } \forall \ell : \cache(\ell) \subseteq \cache'(\ell).
\end{displaymath}
Since we are concerned only with the least cache (the most precise
analysis) we refer to this as {\em the} cache, or synonymously, {\em
the} analysis.

\paragraph{The Algorithm:}

These constraints can be thought of as an abstract evaluator---
$\cache \models t^\ell$ simply means {\em evaluate} $t^\ell$, which
serves {\em only} to update an (initially empty) cache.

\begin{displaymath}
\begin{array}{lcl}
\av{x^\ell}               & = & \cache[\ell \mapsto^+ \cache(x)]\\
\av{(\lambda x.e)^\ell} & = & \cache[\ell \mapsto \{ \lambda x.e \}]\\
\av{(t^{\ell_1}_1\ t^{\ell_2}_2)^\ell} & = & \av{t^{\ell_1}_1};\ \ \av{t^{\ell_2}_2};\\
\ &\ & \mbox{\bf for each }\lambda x.t^{\ell_0}_0 \mbox{\bf\ in } \cache(\ell_1)\mbox{\bf\ do }\\
\ &\ &   \quad   \cache[x \mapsto^+ \cache(\ell_2)];\ \ 
                 \av{t^{\ell_0}_0};\ \
                 \cache[\ell \mapsto^+\cache(\ell_0)]
\end{array}
\end{displaymath}

The abstract evaluator $\av\cdot$ is iterated until the finite cache
reaches a fixed point.\footnote{\label{fn:fineprint}A single iteration
of $\av{e}$ may in turn make a recursive call $\av{e}$ with no change
in the cache, so care must be taken to avoid looping.  This amounts to
appealing to the coinductive hypothesis $\cache \models e$ in
verifying $\cache \models e$.  However, we consider this inessential
detail, and it can safely be ignored for the purposes of obtaining our
main results in which this behavior is never triggered.}  Since the
cache size is polynomial in the program size, so is the running time,
as the cache is {\em monotonic}---we put values in, but never take
them out.  Thus the analysis and any decision problems answered by the
analysis are clearly computable within polynomial time.

\paragraph{An Example:} 

Consider the following program, which we will return to discuss
further in subsequent analyses:

\begin{displaymath}
((\lambda f.((f^1 f^2)^3(\lambda y.y^4)^5)^6)^7 (\lambda
x.x^8)^9)^{10}
\end{displaymath}

The 0CFA is given by the following cache:
\begin{displaymath}
\begin{array}{l@{\:=\:}l@{\hspace{10pt}}l@{\:=\:}l}
\cache(1) & \{ \lambda x \} &
\cache(6) & \{ \lambda x, \lambda y \} \\
\cache(2) & \{ \lambda x \} &
\cache(7) & \{ \lambda f \} \\
\cache(3) & \{ \lambda x, \lambda y \} &
\cache(8) & \{ \lambda x, \lambda y \} \\
\cache(4) & \{ \lambda y \} &
\cache(9) & \{ \lambda x \} \\
\cache(5) & \{ \lambda y \} &
\cache(10)& \{ \lambda x, \lambda y \}
\end{array}
\hspace{10pt}
\begin{array}{l@{\:=\:}l}
\cache(f) & \{ \lambda x \} \\
\cache(x) & \{ \lambda x, \lambda y \}\\
\cache(y) & \{ \lambda y \}
\end{array}
\end{displaymath}
where we write $\lambda x$ as shorthand for $\lambda x.x^8$, etc.

\section{Henglein's simple closure analysis}
\label{sec:simple}

Simple closure analysis follows from an observation by Henglein some
15 years ago: he noted that the standard flow analysis can be computed
in dramatically less time by changing the specification of flow
constraints to use equality rather than containment
\cite{henglein92d}.  The analysis bears a strong resemblance to simple
typing: analysis can be performed by emitting a system of equality
constraints and then solving them using {\em unification}, which can
be computed in almost linear time with a union-find datastructure.

Consider a program with both $(f\;x)$ and $(f\;y)$ as subexpressions.
Under 0CFA, whatever flows into $x$ and $y$ will also flow into the
formal parameter of all abstractions flowing into $f$, but it is not
necessarily true that whatever flows into $x$ {\em also} flows into
$y$ and {\em vice versa}.  However, under simple closure analysis,
this is the case.  For this reason, flows in simple closure analysis
are said to be {\em bidirectional}.

\paragraph{The Analysis:} 

\begin{eqnarray*}
\cache &\models & 
x^\ell \mbox{ iff } \cache(x) = \cache(\ell)\\
\cache &\models & 
(\lambda x .e)^\ell \mbox{ iff } \lambda x.e \in \cache(\ell)\\
\cache &\models & 
(t_1^{\ell_1}\ t_2^{\ell_2})^\ell \mbox{ iff } \cache\models t_1^{\ell_1} \wedge
\cache\models t_2^{\ell_2} \wedge \forall \lambda x.t_0^{\ell_0} \in \cache(\ell_1) : \\
\ & \ & \quad
  \cache\models t_0^{\ell_0} \wedge \cache(\ell_2) = \cache(x) \wedge
  \cache(\ell_0) = \cache(\ell)
\end{eqnarray*}

\paragraph{The Algorithm:} 

We write $\cache[\ell \leftrightarrow \ell']$ to mean
$\cache[\ell\mapsto^+ \cache(\ell')][\ell'\mapsto^+ \cache(\ell)]$.

\begin{displaymath}
\begin{array}{lcl}
\av{x^\ell}               & = & \cache[\ell \leftrightarrow x]\\
\av{(\lambda x.e)^\ell} & = & \cache[\ell \mapsto^+ \{ \lambda x.e \}]\\
\av{(t^{\ell_1}_1 t^{\ell_2}_2)^\ell} & = & \av{t^{\ell_1}_1};\ \ \av{t^{\ell_2}_2};\\
\ &\ & \mbox{\bf for each }\lambda x.t^{\ell_0}_0 \mbox{\bf\ in } \cache(\ell_1)\mbox{\bf\ do }\\
\ &\ &   \quad   \cache[x\leftrightarrow\ell_2];\ \ 
                 \av{t^{\ell_0}_0};\ \
                 \cache[\ell\leftrightarrow\ell_0]
\end{array}
\end{displaymath}
The abstract evaluator $\av{\cdot}$ is iterated until a fixed point is
reached.\footnote{The fine print of footnote \ref{fn:fineprint}
applies as well.}  By similar reasoning to that given for 0CFA, simple
closure analysis is clearly computable within polynomial time.

\paragraph{An Example:}

Recall the example program of the previous section:
\begin{displaymath}
((\lambda f.((f^1 f^2)^3(\lambda y.y^4)^5)^6)^7 (\lambda
x.x^8)^9)^{10}
\end{displaymath}

Notice that $\lambda x.x$ is applied to itself and then to $\lambda
y.y$, so $x$ will be bound to both $\lambda x.x$ and $\lambda y.y$,
which induces an equality between these two terms.  Consequently,
wherever 0CFA gives a flow set of $\{ \lambda x \}$ or $\{ \lambda y
\}$, simple closure analysis will give $\{ \lambda x, \lambda y \}$.
The simple closure analysis is given by the following cache (new flows
are underlined):
\begin{displaymath}
\begin{array}{l@{\:=\:}l@{\hspace{10pt}}l@{\:=\:}l}
\cache(1) & \{ \lambda x, \underline{\lambda y} \} &
\cache(6) & \{ \lambda x, \lambda y \} \\
\cache(2) & \{ \lambda x, \underline{\lambda y} \} &
\cache(7) & \{ \lambda f \} \\
\cache(3) & \{ \lambda x, \lambda y \} &
\cache(8) & \{ \lambda x, \lambda y \} \\
\cache(4) & \{ \lambda y, \underline{\lambda x} \} &
\cache(9) & \{ \lambda x, \underline{\lambda y} \} \\
\cache(5) & \{ \lambda y, \underline{\lambda x} \} &
\cache(10)& \{ \lambda x, \lambda y \}
\end{array}
\hspace{10pt}
\begin{array}{l@{\:=\:}l}
\cache(f) & \{ \lambda x, \underline{\lambda y} \} \\
\cache(x) & \{ \lambda x, \lambda y \}\\
\cache(y) & \{ \lambda y, \underline{\lambda x} \}
\end{array}
\end{displaymath}

\section{Linearity and normalization}
\label{sec:linearity}

In this section, we show that when the program is {\em linear}---every
bound variable occurs exactly once---analysis and normalization are
synonymous.

First, consider an evaluator for our language, $\ev\cdot$:
\begin{displaymath}
\ev\cdot : \Exp \rightarrow \Env \rightharpoonup \langle \Term, \Env\rangle
\end{displaymath}

\begin{displaymath}
\begin{array}{lcl}
\ev{x^\ell}[x \mapsto c]               & = & c\\
\ev{(\lambda x.e)^\ell}\rho & = & \langle\lambda x.e,\rho\rangle \\
\ev{(e_1\ e_2)^\ell}\rho & = &
       \mbox{\bf let }\langle\lambda x.e_0,\rho'\rangle = 
                      \ev{e_1}\rho\restrict\fv{e_1}\mbox{\bf\ in }\\
\ &\ & \mbox{\bf let }c = 
                      \ev{e_2}\rho\restrict\fv{e_2}\mbox{\bf\ in }\\
\ &\ &   \quad   \ev{e_0}{\rho'[x \mapsto c]}
\end{array}
\end{displaymath}
We use $\rho$ to range over {\em environments}, $\Env = \Var
\rightharpoonup \langle\Term,\Env\rangle$, and let $c$ range over {\em
closures}, each comprising a term and an environment that closes the
term.  The function $\lab{\cdot}$ is extended to closures and
environments by taking the union of all labels in the closure or in
the range of the environment, respectively.



Notice that the evaluator ``tightens'' the environment in the case of
an application, thus maintaining throughout evaluation that the domain
of the environment is exactly the set of free variables in the
expression.  When evaluating a variable occurrence, there is only one
mapping in the environment: the binding for this variable. Likewise,
when constructing a closure, the environment does not need to be
restricted: it already is.

In a linear program, each mapping in the environment corresponds to
the single occurrence of a bound variable.  So when evaluating an
application, this tightening {\em splits} the environment $\rho$ into
$(\rho_1,\rho_2)$, where $\rho_1$ closes the operator, $\rho_2$ closes
the operand, and $\dom{\rho_1} \cap \dom{\rho_2} = \emptyset$.


\begin{definition}
Environment $\rho$ {\em linearly closes} $t$ (or $\langle
t,\rho\rangle$ is a {\em linear closure}) iff $t$ is linear, $\rho$
closes $t$, and for all $x\in\dom{\rho}$, $x$ occurs exactly once
(free) in $t$, $\rho(x)$ is a linear closure, and for all
$y\in\dom{\rho}, x$ does not occur (free or bound) in $\rho(y)$. The
{\em size} of a linear closure $\langle t,\rho\rangle$ is defined as:
\begin{eqnarray*}
|t,\rho| & = & |t|+|\rho|\\
|x| & = & 1\\
|(\lambda x.t^\ell)| & = & 1+|t|\\
|(t_1^{\ell_1}\; t_2^{\ell_2})| & = & 1+|t_1|+|t_2|\\
|[x_1\mapsto c_1,\dots,x_n\mapsto c_n]| & = & n+\sum_i |c_i|
\end{eqnarray*}
\end{definition}

The following lemma states that evaluation of a linear closure cannot
produce a larger value.  This is the environment-based analog to the
easy observation that $\beta$-reduction {\em strictly} decreases the
size of a linear term.
\begin{lemma}\label{lem:smaller}
If $\rho$ linearly closes $t$ and $\ev{t^\ell}\rho = c$, then
$|c|\leq|t,\rho|$.
\end{lemma}
\begin{proof}
Straightforward by induction on $|t,\rho|$, reasoning by case analysis
on $t$.  Observe that the size strictly decreases in the application
and variable case, and remains the same in the abstraction case.  \qed
\end{proof}

\begin{definition}
A cache $\cache$ {\em respects} $\langle t,\rho\rangle$ (written $\cache\loves
t,\rho$) when,
\begin{enumerate}
\item $\rho$ linearly closes $t$,
\item $\forall x \in \dom{\rho} . \rho(x) = \langle t',\rho'\rangle
\Rightarrow \cache(x) = \{t'\} \mbox { and } \cache\loves t',\rho'$, and
\item $\forall \ell \in \lab{t}\setminus\fv{t}, \cache(\ell) = \emptyset$.
\end{enumerate}
\end{definition}
Clearly, $\emptyset\loves t,\emptyset$ when $t$ is closed and linear,
i.e.~$t$ is a linear program.

Assume that the imperative algorithm $\av{\cdot}$ of Section
\ref{sec:simple} is written in the obvious ``cache-passing''
functional style.  

\begin{theorem}\label{thm:main}
If $\cache\loves t,\rho$, $\cache(\ell)=\emptyset$,
$\ell\notin\lab{t,\rho}$, $\ev{t^\ell}{\rho} = \langle
t',\rho'\rangle$, and $\av{t^\ell}{\cache} = \cache'$, then
$\cache'(\ell) = \{ t' \}$, $\cache'\loves t',\rho'$, and
$\cache'\models t^\ell$.
\end{theorem}
An important consequence is noted in Corollary \ref{cor:normal}.

\begin{proof} By induction on $|t,\rho|$, reasoning by case analysis on $t$.
\begin{itemize}
\item Case $t\equiv x$.

Since $\cache\loves x,\rho$ and $\rho$ linearly closes $x$, thus $\rho
= [x\mapsto \langle t',\rho'\rangle]$ and $\rho'$ linearly closes
$t'$.  By definition,
\begin{eqnarray*}
\ev{x^\ell}\rho &=& \langle t',\rho'\rangle, \mbox{ and}\\
\av{x^\ell}\cache &=& \cache[x\leftrightarrow \ell].
\end{eqnarray*}
Again since $\cache\loves x,\rho$, $\cache(x) = \{t'\}$, with which the
assumption $\cache(\ell)=\emptyset$ implies
\begin{displaymath}
\cache[x\leftrightarrow\ell](x) =
\cache[x\leftrightarrow\ell](\ell) = 
\{t'\},
\end{displaymath}
and therefore $\cache[x\leftrightarrow\ell]\models x^\ell$.  It
remains to show that $\cache[x\leftrightarrow\ell]\loves t',\rho'$.
By definition, $\cache\loves t',\rho'$.  Since $x$ and $\ell$ do not
occur in $t',\rho'$ by linearity and assumption, respectively, it
follows that $\cache[x\mapsto\ell]\loves t',\rho'$ and the case
holds.

\item Case $t\equiv \lambda x.e_0$.

By definition,
\begin{eqnarray*}
\ev{(\lambda x.e_0)^\ell}\rho & = & \langle\lambda x.e_0,\rho\rangle,\\
\av{(\lambda x.e_0)^\ell}\cache & = & \cache[\ell\mapsto^+ \{\lambda x.e_0\}],
\end{eqnarray*}
and by assumption $\cache(\ell) = \emptyset$, so $\cache[\ell\mapsto^+
\{\lambda x.e_0\}](\ell) = \{\lambda x.e_0\}$ and therefore
$\cache[\ell\mapsto^+ \{\lambda x.e_0\}]\models (\lambda x.e_0)^\ell$.
By assumptions $\ell\notin\lab{\lambda x.e_0,\rho}$ and
$\cache\loves\lambda x.e_0,\rho$, it follows that
$\cache[\ell\mapsto^+ \{\lambda x.e_0\}]\loves\lambda x.e_0,\rho$ and
the case holds.

\item Case $t\equiv t_1^{\ell_1}\; t_2^{\ell_2}$. Let
\begin{eqnarray*}
\ev{t_1}\rho\restrict\fv{t_1^{\ell_1}} &=& \langle v_1,\rho_1\rangle = \langle\lambda x.t_0^{\ell_0},\rho_1\rangle,\\
\ev{t_2}\rho\restrict\fv{t_2^{\ell_2}} &=& \langle v_2,\rho_2\rangle,\\
\av{t_1}\cache &=& \cache_1, \mbox{ and}\\
\av{t_2}\cache &=& \cache_2.
\end{eqnarray*}
Clearly, for $i \in \{1,2\}$, $\cache\loves t_i,\rho\restrict\fv{t_i}$ and
\begin{eqnarray*}
1+\sum_i |t_i^{\ell_i},\rho\restrict\fv{t_i^{\ell_i}}| &=& |(t_1^{\ell_1}\;t_2^{\ell_2}),\rho|.
\end{eqnarray*}

By induction, for $i\in\{1,2\} : \cache_i(\ell_i) = \{v_i\},
\cache_i\loves\langle v_i,\rho_i\rangle,$ and $\cache_i\models
t_i^{\ell_i}$.  From this, it is straightforward to observe that
$\cache_1 = \cache \cup \cache'_1$ and $\cache_2 = \cache \cup
\cache'_2$ where $\cache'_1$ and $\cache'_2$ are disjoint.  So let
$\cache_3 = (\cache_1 \cup \cache_2)[x\leftrightarrow \ell_2]$.  It is
clear that $\cache_3\models t_i^{\ell_i}$.  Furthermore,
\begin{eqnarray*}
\cache_3 &\loves& t_0,\rho_1[x\mapsto \langle v_2,\rho_2\rangle],\\
\cache_3(\ell_0) &=& \emptyset,\mbox{ and}\\
\ell_0 &\notin& \lab{t_0,\rho_1[x\mapsto \langle v_2,\rho_2\rangle]}.
\end{eqnarray*}

By Lemma \ref{lem:smaller},
$|v_i,\rho_i| \leq |t_i,\rho\restrict\fv{t_i}|$, therefore 
\begin{eqnarray*}
|t_0,\rho_1[x\mapsto \langle v_2,\rho_2\rangle]| &<& |(t_1^{\ell_1}\;t_2^{\ell_2})|.
\end{eqnarray*}
Let
\begin{eqnarray*}
\ev{t_0^{\ell_0}}\rho_1[x\mapsto \langle v_2,\rho_2\rangle] &=& \langle v',\rho'\rangle,\\
\av{t_0^{\ell_0}}\cache_3 &=& \cache_4,
\end{eqnarray*}
and by induction, $\cache_4(\ell_0) = \{v'\}$, $\cache_4\loves
v',\rho'$, and $\cache_4\models v'$.  Finally, observe that
$\cache_4[\ell\leftrightarrow\ell_0](\ell) =
\cache_4[\ell\leftrightarrow\ell_0](\ell_0) = \{v'\}$,
$\cache_4[\ell\leftrightarrow\ell_0]\loves v',\rho'$, and
$\cache_4[\ell\leftrightarrow\ell_0]\models
(t_1^{\ell_1}\;t_2^{\ell_2})^\ell$, so the case holds.
\end{itemize}
\qed
\end{proof}
We can now establish the correspondence between analysis and
evaluation.

\begin{corollary}\label{cor:normal}
If $\cache$ is the simple closure analysis of a linear program
$t^\ell$, then $\ev{t^\ell}\emptyset = \langle v,\rho'\rangle$ where
$\cache(\ell) = \{v\}$ and $\cache\loves v,\rho'$.
\end{corollary}

By a simple replaying of the proof substituting the containment
constraints of 0CFA for the equality constraints of simple closure
analysis, it is clear that the same correspondence can be established,
and therefore 0CFA and simple closure analysis are identical for
linear programs.

\begin{corollary}
If $e$ is a linear program, then $\cache$ is the simple closure
analysis of $e$ iff $\cache$ is the 0CFA of $e$.
\end{corollary}

\paragraph{Discussion:}

Returning to our earlier question of the computationally potent
ingredients in a static analysis, we can now see that when the term is
linear, whether flows are directional and bidirectional is irrelevant.
For these terms, simple closure analysis, 0CFA, and evaluation are
equivalent.  And, as we will see, when an analysis is {\em exact} for
linear terms, the analysis will have a \ptime-hardness bound.

\section{Lower bounds for flow analysis}
\label{sec:circuits}

There are at least two fundamental ways to reduce the complexity of
analysis.  One is to compute more approximate answers, the other is to
analyze a syntactically restricted language.

We use {\em linearity} as the key ingredient in proving lower bounds
on analysis.  This shows not only that simple closure analysis and
other flow analyses are \ptime-complete, but the result is rather
robust in the face of analysis design based on syntactic restrictions.
This is because we are able to prove the lower bound via a highly
restricted programming language---the linear $\lambda$-calculus.  So
long as the subject language of an analysis includes the linear
$\lambda$-calculus, and is exact for this subset, the analysis must be
at least \ptime-hard.

The decision problem answered by flow analysis, described colloquially
in Section \ref{sec:intro}, is formulated as follows:
\begin{description}
\item[Flow Analysis Problem:] Given a closed expression $e$, a term
$v$, and label $\ell$, is $v \in \cache(\ell)$ in the analysis of $e$?
\end{description}

\begin{theorem}
If analysis corresponds to evaluation on linear terms, the analysis is
\ptime-hard.
\end{theorem}
The proof is by reduction from the canonical \ptime-complete problem
\cite{ladner-75}:
\begin{description}
\item[Circuit Value Problem:] Given a Boolean circuit $C$ of $n$
inputs and one output, and truth values $\vec{x} = x_1,\dots,x_n$, is
$\vec{x}$ accepted by $C$?
\end{description}

An instance of the circuit value problem can be compiled, using only
logarithmic space, into an instance of the flow analysis problem
following the construction in \cite{vanhorn-mairson-07}.  Briefly, the
circuit and its inputs are compiled into a linear $\lambda$-term,
which simulates $C$ on $\vec{x}$ via {\em evaluation}---it normalizes
to true if $C$ accepts $\vec{x}$ and false otherwise.  But since the
analysis faithfully captures evaluation of linear terms, and our
encoding is linear, the circuit can be simulated by flow analysis.

The encodings work like this: \TT\ is the identity on pairs, and \FF\
is the swap.  Boolean values are either $\langle\TT,\FF\rangle$ or
$\langle\FF,\TT\rangle$, where the first component is the ``real''
value, and the second component is the complement.  
\begin{displaymath}
\begin{array}{rcl@{\qquad}rcl}
\TT & \equiv & \lambda p . \mbox{let }\langle x,y \rangle = p\mbox{ in } \langle x, y\rangle &
\True & \equiv &\langle \TT,\FF \rangle\\
\FF & \equiv & \lambda p . \mbox{let }\langle x,y \rangle = p\mbox{ in } \langle y, x\rangle &
\False & \equiv &\langle \FF,\TT \rangle
\end{array}
\end{displaymath}

The simplest connective is \Not, which is an inversion on pairs, like
\FF.  A {\em linear} copy connective is defined as:
\begin{eqnarray*}
\Copy & \equiv & \lambda b.\mbox{let }\langle u,v\rangle = b\mbox{ in }
\langle u\langle \TT,\FF\rangle, v\langle \FF,\TT\rangle\rangle.
\end{eqnarray*}
The coding is easily explained: suppose $b$ is \True, then $u$ is identity and $v$
twists; so we get the pair $\langle\True,\True\rangle$.  Suppose $b$
is \False, then $u$ twists and $v$ is identity; we get
$\langle\False,\False\rangle$.  

The \And\ connective is defined as:
\begin{displaymath}
\begin{array}{rcl}
\And & \equiv & \lambda b_1.\lambda b_2.\\
\ & \ & \quad\mbox{let }\langle u_1,v_1\rangle = b_1\mbox{ in}\\
\ & \ & \quad\mbox{let }\langle u_2,v_2\rangle = b_2\mbox{ in}\\
\ & \ & \quad\mbox{let }\langle p_1,p_2\rangle = u_1\langle u_2, \FF\rangle \mbox{ in}\\
\ & \ & \quad\mbox{let }\langle q_1,q_2\rangle = v_1\langle \TT, v_2\rangle \mbox{ in}\\
\ & \ & \qquad\langle p_1, q_1 \circ p_2 \circ q_2 \circ \FF \rangle.
\end{array}
\end{displaymath}

Conjunction works by computing pairs $\langle p_1,p_2\rangle$ and
$\langle q_1,q_2\rangle$.  The former is the usual conjuction on the
first components of the Booleans $b_1,b_2$: $u_1\langle u_2,
\FF\rangle$ can be read as ``if $u_1$ then $u_2$, otherwise false
(\FF).''  The latter is (exploiting deMorgan duality) the disjunction of
the complement components of the Booleans: $v_1\langle\TT,v_2\rangle$
is read as ``if $v_1$ (i.e.~if not $u_1$) then true (\TT), otherwise
$v_2$ (i.e.~not $u_2$).''  The result of the computation is equal to
$\langle p_1,q_1\rangle$, but this leaves $p_2,q_2$ unused, which
would violate linearity.  However, there is symmetry to this {\em
garbage}, which allows for its disposal.  Notice that, while we do not
know whether $p_2$ is \TT\ or \FF\ and similarly for $q_2$, we do know
that {\em one of them is \TT\ while the other is \FF}.  Composing the
two together, we are guaranteed that $p_2 \circ q_2 = \FF$.  Composing
this again with another twist (\FF) results in the identity function
$p_2 \circ q_2 \circ \FF = \TT$.  Finally, composing this with $q_1$
is just equal to $q_1$, so $\langle p_1, q_1 \circ p_2 \circ q_2 \circ
\FF \rangle = \langle p_1, q_1\rangle$, which is the desired result,
but the symmetric garbage has been {\em annihilated}, maintaining
linearity.

This hacking, with its self-annihilating garbage, is an improvement
over that given in \cite{mairson-jfp04} and allows Boolean computation
without K-redexes, making the lower bound stronger, but also
preserving all flows.  In addition, it is the best way to do circuit
computation in multiplicative linear logic, and is how you compute
similarly in non-affine typed $\lambda$-calculus.

We know from Corollary \ref{cor:normal} that normalization and
analysis of linear programs are synonymous, and our encoding of
circuits will faithfully simulate a given circuit on its inputs,
evaluating to true iff the circuit accepts its inputs.  But it does
not immediately follow that the circuit value problem can be reduced
to the flow analysis problem.  Let $||C,\vec{x}||$ be the encoding of
the circuit and its inputs.  It is tempting to think the instance of
the flow analysis problem could be stated:
\begin{center}
is \True\ in $\cache(\ell)$ in the analysis
of $||C,\vec{x}||^\ell$?
\end{center}
The problem with this is there may be many syntactic instances of
``\True.''  Since the flow analysis problem must ask about a particular
one, this reduction will not work.  The fix is to use a context which
expects a boolean expression and induces a particular flow (that can
be asked about in the flow analysis problem) iff that expression
evaluates to a true value \cite{vanhorn-mairson-07}.

\begin{corollary}
Simple closure analysis is \ptime-complete.
\end{corollary}

\section{Other monovariant analyses}

In this section, we survey some of the existing monovariant analyses
that either approximate or restrict 0CFA to obtain faster analysis
times.  In each case, we sketch why these analyses are complete for
\ptime.

\subsection{Ashley and Dybvig's sub-0CFA}
\label{sec:sub0}

In \cite{ashley-dybvig}, Ashley and Dybvig develop a general framework
for specifying and computing flow analyses, which can be instantiated
to obtain 0CFA or Jagannathan and Weeks' polynomial $1$CFA
\cite{jagannathan-weeks-95}, for example.  They also develop a class
of instantiations of their framework dubbed {\em sub-0CFA} that is
faster to compute, but less accurate than 0CFA.

This analysis works by explicitly bounding the number of times the
cache can be updated for any given program point.  After this
threshold has been crossed, the cache is updated with a distinguished
$\unknown$ value that represents all possible $\lambda$-abstractions
in the program.  Bounding the number of updates to the cache for any
given location effectively bounds the number of passes over the
program an analyzer must make, producing an analysis that is $O(n)$ in
the size of the program.  Empirically, Ashley and Dybvig observe that
setting the bound to 1 yields an inexpensive analysis with no
significant difference in enabling optimizations with respect to 0CFA.

The idea is the cache gets updated once ($n$ times in general) before
giving up and saying all $\lambda$-abstractions flow out of this
point.  But for a linear term, the cache is only updated at most once
for each program point.  Thus we conclude even when the sub-0CFA bound
is 1, the problem is \ptime-complete.

As Ashley and Dybvig note, for any given program, there exists an
analysis in the sub-0CFA class that is identical to 0CFA (namely by
setting $n$ to the number of passes 0CFA makes over the given
program).  We can further clarify this relationship by noting that for
all linear programs, all analyses in the sub-0CFA class are identical
to 0CFA (and thus simple closure analysis).

\subsection{Subtransitive 0CFA}

Heintze and McAllester \cite{heintze-mcallester-97b} have shown that
the ``cubic bottleneck'' of computing full 0CFA---that is, computing
all the flows in a program---cannot be avoided in general without
combinatorial breakthroughs: the problem is {\sc{2npda}}-hard, for
which the ``the cubic time decision procedure [\dots] has not been
improved since its discovery in 1968.''

Given the unlikeliness of improving the situation in general, Heintze
and McAllester \cite{heintze-mcallester-97a} identify several simpler
flow questions (including the decision problem discussed in the paper,
which is the simplest; answers to any of the other questions imply an
answer to this problem).  They give algorithms for simply typed terms
that answer these restricted flow problems, which under certain
conditions, compute in less than cubic time.

Their analysis is linear with respect to a program's graph, which in
turn, is bounded by the size of the program's type.  Thus, bounding
the size of a program's type results in a linear bound on the running
times of these algorithms.  If this type bound is removed, though, it
is clear that even these simplified flow problems (and their
bidirectional-flow analogs), are complete for \ptime: observe that
every linear term is simply typable, however in our lower bound
construction, the type size is proportional to the size of the circuit
being simulated.  As they point out, when type size is not bounded,
the flow graph may be exponentially larger than the program, in which
case the standard cubic algorithm is preferred.

Independently, Mossin \cite{mossin-98} developed a type-based analysis
that, under the assumption of a constant bound on the size of a
program's type, can answer restricted flow questions such as single
source/use in linear time with respect to the size of the explicitly
typed program.  But again, removing this imposed bound results in
\ptime-completeness.

As Hankin {\em et al.} \cite{hankin-games} point out: both Heintze
and McAllester's and Mossin's algorithms operate on type structure (or
structure isomorphic to type structure), but with either implicit or
explicit $\eta$-expansion.  For simply typed terms, this can result in
an exponential blow-up in type size.  It is not surprising then, that
given a much richer graph structure, the analysis can be computed
quickly.  In this light, recent results on 0CFA of $\eta$-expanded,
simply typed programs can be seen as an improvement of the
subtransitive flow analysis since it works equally well for languages
with first-class control and can be performed with only a fixed number
of pointers into the program structure, i.e.~it is computable in
\logspace\ (and in other words, \ptime\ $=$
\logspace\ up to $\eta$) \cite{vanhorn-mairson-07}.



\section{Conclusions and perspective}


When an analysis is {\em exact}, it will be possible to establish a
correspondence with evaluation.  The richer the language for which
analysis is exact, the harder it will be to compute the analysis.  As
an example in the extreme, Mossin \cite{mossin-97} developed a flow
analysis that is exact for simply typed terms.  The computational
resources that may be expended to compute this analysis are {\em ipso
facto} not bounded by any elementary recursive function
\cite{statman}.  However, most flow analyses do not approach this kind
of expressivity.  By way of comparison, 0CFA only captures \ptime, and
yet researchers have still expending a great deal of effort deriving
approximations to 0CFA that are faster to compute.  But as we have
shown for a number of them, they all coincide on linear terms, and so
they too capture \ptime.



We should be clear about what is being said, and not said.  There is a
considerable difference in practice between linear algorithms
(nominally considered efficient) and cubic algorithms (still feasible,
but taxing for large inputs), even though both are polynomial-time.
\ptime-completeness does not distinguish the two.  But if a
sub-polynomial (e.g., \logspace) algorithm was found for this sort of
flow analysis, it would depend on (or lead to) things we do not know
(\logspace\ $=$ \ptime).  Likewise, were a parallel implementation of
this flow analysis to run in logarithmic time (i.e., \nc), we would
consequently be able to parallelize every polynomial time algorithm
similarly.

A fundamental question we need to be able to answer is this: what can
be deduced about a long-running program with a time-bounded analyzer?
When we statically analyze exponential-time programs with a
polynomial-time method, there should be a analytic bound on what we
can learn at compile-time: a theorem delineating how exponential time
is being viewed through the compressed, myopic lens of polynomial time
computation.

For example, a theorem due to Statman \cite{statman} says this: let
{\bf P} be a property of simply-typed $\lambda$-terms that we would
like to detect by static analysis, where {\bf P} is invariant under
reduction (normalization), and is computable in elementary time
(polynomial, or exponential, or doubly-exponential, or\dots).  Then
{\bf P} is a {\em trivial} property: for any type $\tau$, {\bf P} is
satisfied by {\em all} or {\em none} of the programs of type $\tau$.
Henglein and Mairson \cite{henglein-mairson-popl91} have complemented
these results, showing that if a property is invariant under
$\beta$-reduction for a class of programs that can encode all Turing
Machines solving problems of complexity class {\sc{f}} using
reductions from complexity class {\sc{g}}, then any superset is either
{\sc{f}}-complete or trivial.  Simple typability has this property for
linear and linear affine $\lambda$-terms
\cite{henglein-mairson-popl91,mairson-jfp04}, and these terms are
sufficient to code all polynomial-time Turing Machines.

We would like to prove some analogs of these theorems, with or without
the typing condition, but weakening the condition of ``invariant under
reduction'' to some {\em approximation} analogous to the
approximations of flow analysis, as described above.  We are motivated
as well by yardsticks such as Shannon's theorem from information
theory \cite{shannon}: specify a bandwidth for communication and an
error rate, and Shannon's results give bounds on the channel capacity.
We too have essential measures: the time complexity of our analysis,
the asymptotic differential between that bound and the time bound of
the program we are analyzing.  There ought to be a fundamental result
about what information can be yielded as a function of that
differential.  At one end, if the program and analyzer take the same
time, the analyzer can just run the program to find out everything.
At the other end, if the analyzer does no work (or a constant amount
of work), nothing can be learned.  Analytically speaking, what is in
between?

\paragraph{Acknowledgments:}

We are grateful to Olin Shivers and Matt Might for a long, fruitful,
and ongoing dialogue on flow analysis.  We thank the anonymous
reviewers for insightful comments.  The first author also thanks the
researchers of the Northeastern University Programming Research Lab
for the hospitality and engaging discussions had as a visiting
lecturer over the last year.

\bibliography{simpleclosure}
\end{document}